\documentclass[11pt, fleqn, a4paper]{scrartcl}
\usepackage{ifthen}
\newboolean{IsProc}\setboolean{IsProc}{false}
\usepackage[english]{babel}
\usepackage[utf8]{inputenc}
\usepackage[T1]{fontenc}
\usepackage{amssymb}
\usepackage{algorithm}
\usepackage[noend]{algpseudocode}
\renewcommand{\algorithmicrequire}{\textbf{Input:}}

\algnewcommand{\LineComment}[1]{\hspace{-1.6em}\(\triangleright\) #1}
\usepackage{color}
\ifthenelse{\boolean{IsProc}}{
\usepackage[lmargin=1.0in,rmargin=1.0in,bottom=1.0in,top=1.0in,
twoside=False]{geometry}}{
\usepackage[lmargin=1.1in,rmargin=1.1in,bottom=1.3in,top=1.3in,
twoside=False]{geometry}}
\usepackage{enumerate}
\usepackage{relsize,xspace}
 \usepackage{xcolor}
 \usepackage{mathtools}
\usepackage[normalem]{ulem}
\usepackage{microtype}
\usepackage{amsmath}
\usepackage{amssymb}
\usepackage{amsfonts}
\usepackage{stmaryrd}
\usepackage{times}
\usepackage{tikz}
\usepackage{bm}
\usepackage{tikz}
\tikzstyle{vertex}=[circle,inner sep=2.5,minimum size
=2mm,semithick,fill=black!20, draw=black]
\usetikzlibrary{decorations.pathreplacing}
\tikzstyle{brace}=[thin,decorate,decoration=brace]

\definecolor{blue}{rgb}{0.1,0.2,0.5}
\definecolor{brown}{rgb}{0.6,0.6,0.2}
\usepackage[ocgcolorlinks, linkcolor={blue}, citecolor={brown}]{hyperref}

\usepackage[amsmath,thmmarks,hyperref]{ntheorem}
\usepackage{cleveref}

\crefformat{page}{#2page~#1#3}%
\Crefformat{page}{#2Page~#1#3}%
\crefformat{equation}{#2(#1)#3}%
\Crefformat{equation}{#2(#1)#3}%
\crefformat{figure}{#2Figure~#1#3}%
\Crefformat{figure}{#2Figure~#1#3}%
\crefformat{section}{#2Section~#1#3}
\Crefformat{section}{#2Section~#1#3}
\crefformat{chapter}{#2Chapter~#1#3}
\Crefformat{chapter}{#2Chapter~#1#3}
\crefformat{chapter*}{#2Chapter~#1#3}
\Crefformat{chapter*}{#2Chapter~#1#3}
\crefformat{part}{#2Part~#1#3}
\Crefformat{part}{#2Part~#1#3}
\crefformat{enumi}{#2(#1)#3}
\Crefformat{enumi}{#2(#1)#3}
\crefformat{algorithm}{#2algorithm~#1#3}%
\Crefformat{algorithm}{#2Algorithm~#1#3}%

\usepackage{latexsym}

\theoremnumbering{arabic}

\theoremstyle{plain}
\theoremsymbol{}
\theorembodyfont{\itshape}
\theoremheaderfont{\normalfont\bfseries}
\theoremseparator{}

\theoremseparator{.}
\newtheorem{theorem}{Theorem}
\crefformat{theorem}{#2Theorem~#1#3}
\Crefformat{theorem}{#2Theorem~#1#3}

\newcommand{\newtheoremwithcrefformat}[2]{%
  \newtheorem{#1}[theorem]{#2}%
  \crefformat{#1}{##2\MakeUppercase#1~##1##3}%
  \Crefformat{#1}{##2\MakeUppercase#1~##1##3}%
}

\newtheoremwithcrefformat{proposition}{Proposition}
\newtheoremwithcrefformat{observation}{Observation}
\newtheoremwithcrefformat{lemma}{Lemma}
\newtheoremwithcrefformat{conjecture}{Conjecture}
\newtheoremwithcrefformat{corollary}{Corollary}
\theorembodyfont{\upshape}
\newtheoremwithcrefformat{example}{Example}
\newtheoremwithcrefformat{remark}{Remark}

\theoremsymbol{\ensuremath{□}}
\newtheoremwithcrefformat{definition}{Definition}

\theoremstyle{nonumberplain}
\theoremseparator{}
\theoremheaderfont{\scshape}
\theorembodyfont{\normalfont}
\theoremsymbol{\ensuremath{\square}}
\newtheorem{proof}{Proof.}

\DeclarePairedDelimiter{\abs}{\lvert}{\rvert}

\newcommand{\wcol}{\mathrm{wcol}}

\newcommand{\Wreach}{\mathrm{WReach}}

\newcommand{\Oof}{\mathcal{O}}
\newcommand{\CCC}{\mathcal{C}}
\newcommand{\XXX}{\mathcal{X}}
\newcommand{\minor}{\preccurlyeq}
\newcommand{\N}{\mathbb{N}}
\newcommand{\dist}{\mathrm{dist}}
\newcommand{\rad}{\mathrm{rad}}

\renewcommand{\mid}{~:~}
\renewcommand{\epsilon}{\varepsilon}

\newcommand{\congestbc}{\ensuremath{\mathcal{CONGEST}_{\hspace{-2pt}\textsc{BC}}}\xspace}
\newcommand{\congest}{\ensuremath{\mathcal{CONGEST}}\xspace}
\newcommand{\local}{\ensuremath{\mathcal{LOCAL}}\xspace}
\usepackage{authblk}
\newcommand*\samethanks[1][\value{footnote}]{\footnotemark[#1]}

\title{Distributed Domination on Graph Classes of Bounded Expansion}
\date{\vspace{-5ex}}
\ifthenelse{\boolean{IsProc}}{\subtitle{Regular Submission}}{}
\author[1]{Saeed Akhoondian Amiri%
\footnote{S.A. Amiri, R. Rabinovich and, partially, S.
  Siebertz's research was supported by the
European Research Council (ERC) under the European Union's Horizon
2020 research and innovation programme (ERC consolidator grant DISTRUCT,
agreement No 648527).}%
}
\author[2]{Patrice Ossona de Mendez
\thanks{P. Ossona de Mendez's research was supported by
grant ERCCZ LL-1201 and by the European Associated Laboratory ``Structures in
Combinatorics'' (LEA STRUCO), and partially supported by ANR project Stint under
reference ANR-13-BS02-0007}%
}
\author[1]{Roman~Rabinovich\samethanks[1]}%
\author[3]{Sebastian Siebertz
\thanks{Contact author. Address: Institute of Informatics,
Faculty of Mathematics, Informatics, and Mechanics of the University of Warsaw
ul. Banacha 2, 02-097 Warsaw, Poland. Telephone number: +48 22 5544458. \\The work of S. Siebertz is partially supported by the National
  Science Centre of Poland via POLONEZ grant agreement UMO-2015/19/P/ST6/03998, 
which has received funding from the European Union's Horizon 2020 research and 
innovation programme (Marie Sk\l odowska-Curie grant agreement No.\ 665778).
}%
}
\affil[ ]{\normalsize\texttt{saeed.amiri@tu-berlin.de, pom@ehess.fr,
    roman.rabinovich@tu-berlin.de, siebertz@mimuw.edu.pl}}
\affil[ ]{}
\affil[1]{Technische Universität Berlin, Germany}
\affil[2]{CAMS (CNRS UMR 8557)\\
  EHESS, Paris, France}
\affil[3]{Institute of Informatics, University of Warsaw, Poland}

\begin{document}

\maketitle

\begin{picture}(0,0) \put(385,-410)
{\hbox{\includegraphics[scale=0.25]{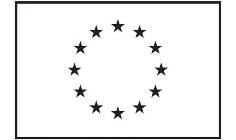}}} \end{picture} 
\vspace{-0.8cm}
\begin{abstract}
\begin{abstract}
  \noindent We provide a new constant factor approximation algorithm
  for the (connected) \mbox{distance-$r$} dominating set problem on graph classes of
  bounded expansion. Classes of bounded expansion include many
  familiar classes of sparse graphs such as planar graphs and graphs
  with excluded (topological) minors, and notably, these classes 
  form the most general subgraph closed classes of graphs for
  which a sequential constant factor approximation algorithm for the
  \mbox{distance-$r$} dominating set problem is currently known. Our 
  algorithm can be implemented in the \congestbc model of 
  distributed computing and uses $\Oof(r^2 \log n)$ communication 
  rounds. 

  Our techniques, which may be of independent interest, are based on
  a distributed computation of sparse neighborhood covers of small
  radius on bounded expansion classes. We show how to compute 
  an $r$-neighborhood cover of radius~$2r$ and
  overlap $f(r)$ on every class of bounded expansion in $\Oof(r^2\log n)$
  communication rounds  for some function~$f$.

  Finally, we show how to use the greater power of the \local model 
  to turn any distance-$r$ dominating set 
  into a constantly larger connected distance-$r$ dominating
  set in $3r+1$ rounds on any class of bounded expansion. 
  Combining this algorithm,
  e.g., with the constant factor approximation algorithm for dominating 
  sets on planar graphs of Lenzen et al.\ 
  gives a constant factor approximation algorithm for connected 
  dominating sets on planar graphs in a constant number of rounds in 
  the \local model, where the approximation ratio is only $6$ times larger
  than that of Lenzen et al.'s algorithm. 
\end{abstract}
	
\end{abstract}

\section{Introduction and contributions}

The \textsc{Dominating Set} and \textsc{Connected Dominating Set} problems are two of the most well-studied problems in algorithms and combinatorics \cite{haynes1998fundamentals}.
Recall that 
a subset $D$ of vertices of a graph~$G$ is a {\em dominating set}  of $G$ if every vertex of $G$ is either in $D$ or adjacent to a vertex in $D$, and that a dominating set is a {\em connected dominating set} if it induces a connected subgraph.

The \textsc{Dominating Set}  problem, which aims at 
finding a minimum size dominating set in a graph 
is \textsc{NP}-complete in general~\cite{karp1972reducibility}, 
and even so on planar graphs of maximum degree $3$ (cf.\ [GT2] in~\cite{michael1979computers}). 
The simple greedy algorithm---which 
 at each step  adds a vertex dominating the largest number of non-dominated vertices---achieves an approximation ratio
\footnote{Note that these results are for the \textsc{Set Cover} problem, 
which however reduces to the \textsc{Dominating Set} problem
by an approximation preserving reduction and,
in fact, the two problems achieve exactly the same approximation 
ratio~\cite{kann1992approximability}.}  of 
$\ln n -\ln \ln n + \Theta(1)$ on graphs of order $n$~\cite{chvatal1979greedy,
lovasz1975ratio}, and
 no better approximation ratio can be achieved in general under
standard complexity theoretic assumptions~\cite{alon2006algorithmic,
arora2003improved, bellare1993efficient, dinur2014analytical, 
feige1998threshold, lund1994hardness, raz1997sub}.

Recall that a class $\mathcal C$ (closed under taking subgraphs) has {\em bounded expansion} if, for every integer~$r$, the average degree of all graphs having their $r$-subdivision (that is the graph obtained by replacing each edge by a path of length $r+1$)  in $\mathcal C$ is bounded by some constant $C(r)$ \cite{nevsetvril2008grad, nevsetvril2008gradb,
nevsetvril2008gradc,NesetrilOdM12}.
Not only do many familiar classes of sparse graphs, such as planar graphs and graphs with excluded (topological) minors  have bounded expansion, but so do many geometrically defined graphs \cite{BEEx, har2015approximation},  and 
experimental evidence that real world complex networks  have bounded expansion is given in~\cite{demaine2014structural}. 
Also, widely used random models of sparse graphs, like the Configuration Model \cite{molloy1995critical} and the Chung--Lu Model \cite{chung2002connected} with specified asymptotic degree sequences 
generate graphs that asymptotically almost surely belong to a bounded expansion class determined by the parameters of the model~\cite{demaine2014structural}.

In \cite{dvovrak13} it is proved that in linear time one can 
compute a constant factor approximation on classes with bounded expansion for a generalization of the \textsc{Dominating Set}, the \textsc{Distance-$r$ Dominating Set} problem, which consists in finding in an input graph $G$ a minimum size subset $D$ of vertices, such that every vertex of $G$ is at distance at most $r$ from a vertex in $D$. 
This problem (also known 
as the \textsc{$(k,r)$-center}
problem)
 has been extensively studied in the literature.


\ifthenelse{\boolean{IsProc}}%
{\bigskip}%
{\bigskip}
\noindent \textbf{Contribution 1.} We present a new approximation algorithm for 
the \textsc{Distance-$r$ Dominating Set} problem on classes of 
bounded expansion which improves the approximation ratio achieved 
by the algorithm of~\cite{dvovrak13}. Our algorithm can be
implemented in linear time on any class of bounded expansion. A key
feature of our algorithm is that it is tailored to be executed
in a distributed setting.

\ifthenelse{\boolean{IsProc}}%
{\medskip}%
{\bigskip}
There has been lots of effort to approximate the \textsc{Dominating Set} problem with
distributed algorithms, however, similar hardness results
also apply to distributed algorithms. 
It was shown in~\cite{kuhn2016local} 
that in~$t$ communication rounds the 
\textsc{Dominating Set} problem on $n$-vertex graphs of maximum degree 
$\Delta$  can only be approximated within factor~$\Omega(n^{c/t^2})$ and~$\Omega(\Delta^{c'/t})$, where~$c$ and~$c'$ are constants. 
This implies that, in general, to achieve a constant approximation ratio, 
every distributed algorithm requires at least~$\Omega(\sqrt{\log n})$ and~$\Omega(\log \Delta)$ communication rounds. 
Kuhn et al.~\cite{kuhn2016local} also provides an approximation algorithm on general graphs, which achieves a~$(1+\epsilon)\ln \Delta$-approximation in~$\Oof(\log(n)/\epsilon)$ rounds for any~$\epsilon>0$. 
Ghaffari et al.~\cite{GhaffariKM17} provide a polylog-time distributed approximation scheme for covering and packing integer linear programs.
In particular, based on their techniques one can compute a $(1+\epsilon)$-approximation for dominating sets and distance-$r$ dominating sets
in $\Oof(\mathit{poly}(\log n/\epsilon))$ rounds in the \local model. 
Observe however that this result requires learning 
polylog-neighborhoods and solving the dominating set problem optimally
on these neighborhoods. In particular, the algorithm cannot be 
carried out in the \congest model (in this number of steps).

For graphs of arboricity~$a$ there exists a
forest decomposition algorithm achieving a factor~$\Oof(a^2)$-approximation
in randomized time~$\Oof(\log n)$, and a deterministic~$\Oof(a \log
\Delta)$ approximation algorithm
requiring~$\Oof(\log \Delta)$ rounds~\cite{lenzen2010minimum}. 
Given any~$\delta>0$,~$(1+\delta)$-approximations of a maximum independent
set, of a maximum matching, and of a minimum dominating set can be computed in
$\Oof(\log^* n)$ rounds in planar graphs~\cite{czygrinow2008fast}, which 
is asymptotically optimal~\cite{lenzen2008leveraging}. 
A constant factor
approximation on planar graphs~\cite{lenzen2013distributed,wawrzyniak2014strengthened} 
and on graphs of bounded genus~\cite{AmiriSS16} can be computed locally in a
constant number of communication rounds. 
In terms of lower bounds, it was shown that there is no 
deterministic local algorithm (constant-time distributed graph algorithm) that 
finds a~$(7-\epsilon)$-approximation of a minimum dominating set on planar graphs, 
for any positive constant~$\epsilon$~\cite{hilke2014brief}. 

Observe that
the above algorithms for restricted graph classes cannot be directly employed to obtain good 
approximations for the \textsc{Distance-$r$ Dominating Set} problem, 
as all structural information which is used in the algorithms may be 
lost when building the $r$-transitive closure of the graph. The distributed 
algorithms of~\cite{kutten1995fast,penso2004distributed} find distance-$r$
dominating sets of size $\Oof(n/r)$ in time $\Oof(r\cdot \log^*n)$, without 
any relation to the size of an optimal distance-$r$ dominating set. In very
restrictive settings, e.g., in trees~\cite{turaudistributed} or in star-split 
graphs~\cite{wang2003distributed} better solutions are known.

%
The \textsc{Distance-$r$ Dominating Set} problem is closely related
to the problem of covering local neighborhoods in a graph by connected 
clusters of small radius. 
An 
$r$-neigh\-bor\-hood cover~\cite{awerbuch1990sparse} is a set $\XXX$ 
of vertex sets $X\subseteq V(G)$ such that for each vertex $v\in V(G)$ 
there is a set $X\in \XXX$ with $N_r[v]\subseteq X$. We are
interested in covers of small \emph{radius}, that is, $\rad(G[X])$ shall
be small for all $X\in \XXX$ and small \emph{degree}, that is, 
every vertex $v\in V(G)$ shall lie in only a few clusters. 

Sparse covers have many applications such as distance coordinates,
routing with succinct routing tables~\cite{abraham2005compact,awerbuch1990sparse}, mobile user tracking~\cite{awerbuch1990sparse}, 
resource allocation~\cite{awerbuch1990locality}, synchronisation
in distributed algorithms~\cite{awerbuch1990network}, and many more.
Every graph admits an $r$-neighborhood cover of radius at most $2r-1$ 
and degree at most $2r\cdot n^{1/r}$~\cite{awerbuch1990sparse} and
asymptotically, this cannot be improved~\cite{thorup2005approximate}. 
Better covers are known to exist, e.g., for planar graphs~\cite{busch2007improved} and for classes
that exclude a minor~\cite{abraham2007strong}. In particular, the construction 
of~\cite{abraham14} provides $r$-neighborhood covers of radius $\Oof(t^2r)$ and degree $2^{\Oof(t)}\cdot t!$
for graphs
that exclude $K_t$ as a minor. It follows from a construction in~\cite{grohe2014deciding} that classes of bounded expansion admit 
$r$-neighborhood covers of radius at most $2r$ and degree at most
$f(r)$ for some function $f$.

\ifthenelse{\boolean{IsProc}}%
{\bigskip}%
{\medskip}
\noindent\textbf{Contribution 2.} We show that the algorithm
of~\cite{grohe2014deciding} for constructing sparse
$r$-neighborhood covers on classes of bounded expansion 
can be implemented in the 
\congestbc model of distributed computing in $\Oof(r^2\log n)$ 
communication rounds. Based on this construction, we
show that our newly proposed algorithm for the \textsc{Distance-$r$ Dominating
Set} problem can be implemented in the \congestbc model in 
$\Oof(r^2\log n)$ communication rounds on any class of graphs of bounded expansion.
Our result is based on a routing scheme presented by Ne\v{s}et\v{r}il
and Ossona de Mendez in~\cite{nevsetvrildistributed}, which in turn is
based on an iterative application of an algorithm of Barenboim and 
Elkin~\cite{barenboim2010sublogarithmic}. 

\ifthenelse{\boolean{IsProc}}%
{\medskip}%
{\bigskip}
While in a sequential setting one can trivially connect the vertices of 
a (distance-$r$) dominating set along a spanning tree to obtain
a connected (distance-$r$) dominating set of small size, creating 
such connections is a non-trivial task in the distributed setting. 
Several algorithms were proposed to compute connected dominating
sets in general graphs~\cite{das1997routing,das1997routingb,kuhn2016local,sivakumar1998improved,
stojmenovic2002dominating,wu1999calculating}. We also refer to 
these papers for applications of connected dominating
sets for distributed computing and routing. All lower bounds for the 
\textsc{Dominating Set} problem hold all the more so for the
\textsc{Connected Dominating Set} problem. In particular, none
of the above algorithms computes a constant factor approximation
of a minimum connected dominating set
in a sub-linear number of communication rounds. 
To our knowledge, there is no distributed 
algorithm to compute a constant factor approximation to the \textsc{Connected (Distance-$r$) Dominating Set} problem on restricted graph classes. 

\ifthenelse{\boolean{IsProc}}%
{\medskip}%
{\bigskip}
\noindent\textbf{Contribution 3.} 
We show how to extend our algorithm for the \textsc{Distance-$r$ 
Dominating Set} problem to compute a constant factor approximation for
the \textsc{Connected Distance-$r$ Dominating Set} problem. We hence prove
that there exists a constant
factor approximation algorithm for the \textsc{Connected Distance-$r$ Dominating Set} 
problem which works in the \congestbc model in 
$\Oof(r^2\log n)$ communication rounds on any class of graphs of bounded expansion.

Finally, we show how to use the greater power of the \local model to 
turn any distance-$r$ 
dominating set $D$ into a connected distance-$r$
dominating set of size at most $c(r)\cdot |D|$, for some small constant 
$c(r)$ depending only on $r$ and the class under consideration. This
new algorithm can be implemented in $3r+1$ communication rounds
in the \local model. In combination with the algorithm 
of Lenzen et al.~\cite{lenzen2013distributed} we obtain a constant factor
approximation algorithm for the \textsc{Connected Dominating
Set} problem on planar graphs in a constant number of
communication rounds in the \local model (the constant 
$c(1)$ which we need here is $6$). A similar result follows
for graphs of bounded genus by combining our new algorithm with an algorithm 
of~\cite{AmiriSS16}.





\section{Preliminaries}\label{sec:prelim}

\subparagraph*{Graphs.} In this paper, we consider \emph{finite, undirected simple graphs}. 
For a graph $G$, we write $V(G)$ for the \emph{vertex set} of $G$ and $E(G)$ 
for its \emph{edge set}. A \emph{path of length $\ell$} in $G$ is a subgraph 
$P\subseteq G$ with vertex set $V(P)=\{v_1,\ldots, v_{\ell+1}\}$ and edge set 
$E(P)=\{\{v_i,v_{i+1}\} \mid  1\leq i<\ell\}$. The path~$P$ \emph{connects its 
endpoints} $v_1$ and $v_{\ell+1}$. The \emph{distance} between two vertices 
$u,v\in V(G)$, denoted $\dist(u,v)$, is the minimum length of a path that connects 
$u$ and $v$ or $\infty$ if no such path exists. For $v\in V(G)$, we write $N_r[v]$ 
for the \emph{closed $r$-neighborhood of $v$}, that is $N_r[v]=\{u\in V(G) \mid  \dist(u,v)\leq r\}$.
Note that we allow paths of length $0$, so $N_r[v]$ always contains $v$ itself. 
For a set $A\subseteq V(G)$, we write $N_r[A]$ for $\bigcup_{v\in A}N_r[v]$. 
The radius of a connected graph $G$ is the minimum number $\rad(G)$ such that there is a 
vertex $v\in V(G)$ with $N_{\rad(G)}[v]=V(G)$. 


The \emph{arboricity} of a graph is the minimum number of spanning
forests that partition its edge set. The arboricity of a graph is
within factor $2$ of its degeneracy.
For a set $X\subseteq V(G)$ we write $G[X]$ for the subgraph of $G$ \emph{induced} 
by $X$. For $k\in \N$, $G$ is \emph{$k$-degenerate} if for each $X\subseteq V(G)$ the graph $G[X]$ contains a vertex of
degree at most~$k$. If an $n$-vertex graph $G$ is $k$-degenerate, then $G$ contains at most $k\cdot n$ edges.

We assume that all graphs are represented by adjacency lists\ifthenelse{\boolean{IsProc}}{}{ so that 
the total size of a graph representation is linear in the number 
of edges and vertices}.\ifthenelse{\boolean{IsProc}}{}{

} If~$G$ is $k$-degenerate, then in linear time we can order the vertices as
$v_1,\ldots, v_n$ such that every vertex $v_i$ has at most~$k$ smaller
neighbours $v_{j_1},\ldots, v_{j_k}$, $j_\ell < i$ for
$\ell\in \{1,\ldots, k\}$. In the same time complexity we can order all
adjacency lists consistently with the order. For the sequential model in
\Cref{sec:approx} we assume that the identifiers of the vertices occupy
constant space while in the distributed model we assume $\log n$-bit
identifiers.

\subparagraph*{Distance\hspace{2pt}-\hspace{1pt}\textit{r} dominating sets.} For an integer $r$, a
\emph{distance\hspace{1pt}-\hspace{1pt}$r$ dominating set} in a graph $G$ is a 
set $M\subseteq V(G)$ such that $N_r[M]=V(G)$. A distance-$1$ dominating set is simply called a 
\emph{dominating set}. 

\subparagraph*{Distributed system model.} The clients of a network are modelled as the vertices 
$V(G)$ of a graph~$G$, its communication links are represented by the edges 
$E(G)$ of the graph. Each client has a unique identifier \textit{(id)} of size
$\log n$ where $n\coloneqq |V(G)|$ is the order of the graph known to every vertex. 
Communication 
is synchronous and reliable. In each round, each vertex $v\in V(G)$ may send a 
(different) message to each of its neighbors $w\in N_1[v]$ (the vertex specifies 
which message is sent to which neighbor) and receives all messages from its 
neighbors. In the \local model, messages may have arbitrary size, 
in the \congest model, messages may have size $\Oof(\log n)$. 
In the \congestbc model,
every vertex may only broadcast the same message of size $\Oof(\log n)$ to all its
neighbors. 
After sending and receiving messages, every client may perform arbitrary finite 
computations. 
The complexity of a distributed algorithm is its number of
communication rounds. 
The network graph also represents the graph problem that we are trying to solve, 
e.g., the \textsc{\mbox{Distance-$r$} Dominating Set} instance. At termination, each vertex must output whether it is part of the 
\mbox{distance-$r$} dominating set or not, and these outputs must define a valid solution of the 
problem. 
We refer to~\cite{peleg00} for more background.

\subparagraph*{Bounded expansion classes.}
A graph $H$ with vertex set $V(H)=\{v_1,\ldots, v_n\}$ is a \emph{minor} of~$G$, 
written $H\minor G$, if there are pairwise disjoint connected subgraphs 
$H_1,\ldots, H_n\subseteq G$, called {\em{branch sets}},
such that if $\{v_i,v_j\}\in E(H)$, then there are vertices $u_i\in V(H_i)$ and 
$u_j\in V(H_j)$ with $\{u_i,u_j\}\in E(G)$. We call $(H_1,\ldots, H_n)$ a 
{\em{minor model}} of $H$ in~$G$. For $r\in \N$, the graph $H$ 
is a {\em{depth-$r$ minor}} of~$G$, denoted
$H\minor_rG$, if there is a minor model $(H_1,\ldots,H_n)$ of $H$ in~$G$ 
such that each~$H_i$ has radius at most $r$. We write $d(H)$ for the 
\emph{average degree} of~$H$, that is, for the number $2\abs{E(H)}/\abs{V(H)}$.
A class $\CCC$ of graphs has \emph{bounded expansion} if there is a function
$f\colon\N\rightarrow\N$ such that for all $r\in \N$ and all
graphs $H$ and $G\in \CCC$, if $H\minor_r G$, then
$d(H)\leq f(r)$. 
Observe that every $n$-vertex graph from a 
bounded expansion class is $f(0)$-degenerate and hence has at most $f(0)\cdot n$
many edges (depth-$0$ minors of $G$ are its subgraphs). 

\subparagraph*{Generalized colouring numbers.}
Let $G$ be a graph. A \emph{linear order} $L$ of $V(G)$ is a reflexive, anti-symmetric,
transitive total binary relation 
$L\subseteq V(G)\times V(G)$. In the following, we will write $u\leq_Lv$ 
instead of $(u,v)\in L$. We write $\Pi(G)$ for the set of all linear orders on $V(G)$. 
Let $r\in \N$ and let $u,v\in V(G)$. Vertex $u$ is \emph{weakly 
$r$-reachable} from vertex $v$ with respect to a linear order $L\in \Pi(G)$ if there 
exists a path $P$ of length at most~$r$ between $u$ and $v$ such that 
$u$ is minimum among the vertices of $P$ (with respect to $L$). Let 
$\Wreach_r[G, L, v]$ be the set of vertices that are weakly $r$-reachable 
from~$v$ with respect to~$L$. Note that $v\in \Wreach_r[G,L,v]$. The 
\emph{weak $r$-colouring number} $\wcol_r(G)$ of $G$ is defined as
\begin{align*}
  \wcol_r(G)=\min_{L\in\Pi(G)}\max_{v\in V(G)}\abs{\Wreach_r[G,L, v]}.
\end{align*}
The generalized colouring numbers were introduced by Kierstead and Yang in the 
context of colouring games and marking games on graphs \cite{kierstead03}, and 
received much attention as a measure for uniform sparseness in graphs, in particular, 
they can be used to characterize classes of bounded expansion. 

\begin{theorem}[Zhu \cite{zhu2009colouring}]\label{thm:zhu}
A class $\CCC$ of graphs has bounded expansion if and only if there is a function 
$f\colon \N\rightarrow\N$ such that $\wcol_r(G)\leq f(r)$ for all $r\in \N$. 
\end{theorem}

Bounds for several restricted classes such as graphs 
of bounded tree-width, planar graphs or graphs with excluded (topological)
minors, were
provided in~\cite{GroheKRSS15,KreutzerPRS16,van2015generalised}. 
\newcounter{lemma-orderfromaugmentation-new}
\setcounter{lemma-orderfromaugmentation-new}{\value{theorem}}
\newcommand{\prfOrderfromaugmentationNew}{}%
The weak colouring numbers can be well approximated in linear
time as shown in \cite{dvovrak13}. 

\begin{theorem}[Dvo\v{r}\'ak~\cite{dvovrak13}]\label{crl:crl-orientation}
Let $\CCC$ be a class of bounded expansion. There is a linear time
algorithm and a function $d:\N\rightarrow \N$ which on input
$G\in \CCC$ and $r\in \N$ computes in linear time 
a linear order of $V(G)$ witnessing that
$\wcol_r(G)\leq d(r)$. 
\end{theorem}

%
The next theorem is implicit in \cite[Section 4.4]{nevsetvrildistributed} and shows that we can compute these orders also
in the distributed setting. 

%

\begin{theorem}[Ne\v{s}et\v{r}il and Ossona de Mendez \cite{nevsetvrildistributed}]\label{crl:compute-distributed-wcol}
Let $\CCC$ be a class of bounded expansion and let $r\in \N$. There is 
a constant $d(r)$ such that one can compute for every $G\in \CCC$
in $\Oof(r^2\log n)$ communication rounds in the \congestbc model
an order of $V(G)$ witnessing that \mbox{$\wcol_r(G)\leq d(r)$}. 
\end{theorem}

The procedure described in~\cref{crl:compute-distributed-wcol} uses an algorithm by
Barenboim and Elkin~\cite{barenboim2010sublogarithmic} 
which computes an orientation of degenerate graphs. 
For this, we must assume that all vertices know the order $n$ of the input
graph. The order is represented by assigning every vertex a \emph{class-id}, 
which together with the unique vertex-id induces a total order of $V(G)$. 
We remark that (though not explicitly stated) this order can be obtained
 as a by-product of the procedure 
$\mathrm{orient}(z,C)$ described in \cite[Section 4.4]{nevsetvrildistributed}. 

In fact, it suffices if the vertices know a polynomial
approximate of $n$ as this changes $\log n$ only by a constant
factor. In the following, we will assume that the exact value of $n$ 
is available to avoid unnecessary complication.

%
%

\subparagraph*{Sparse neighborhood covers.} Let $G$ be a graph. For $r\in \N$, an 
\emph{$r$-neighborhood cover} of $G$ is a set~$\XXX$ of subsets $X\subseteq V(G)$, called the {\em clusters} of $\XXX$, such that for each $v\in V(G)$ there is some 
$X\in \XXX$ with $N_r(v)\subseteq X$. The \emph{radius} of $\XXX$ is the 
maximum radius of the graph induced by a cluster $X\in\XXX$. 
Note that in every $r$-neighborhood 
cover of bounded radius each of the clusters induces a connected subgraph of~$G$. 
The \emph{degree} $d_\XXX(v)$ of a vertex 
$v\in V(G)$ with respect to $\XXX$ is the number of clusters that contain~$v$. 
The degree of~$\XXX$ is the maximum degree $d_\XXX(v)$ over all vertices 
$v\in V(G)$. The generalized colouring numbers can be used to construct 
sparse neighborhood covers. 

We fix a number $r\in \N$ for the remainder of
the paper. 
For a vertex $v\in V(G)$, let
$X_v$ be the set of the vertices $w$ such that $ v\in \Wreach_{2r}[G,L,w]$.
\begin{theorem}[Grohe et al.~\cite{grohe2014deciding}]\label{thm:nc}
Let $G$ be a graph and let $c, r\in \N$. Let~$L$ be an order witnessing that 
$\wcol_{2r}(G)\leq c$. Then the collection $\mathcal{X} = \{X_v \mid  v\in V(G)\}$ 
is an $r$-neighborhood cover of~$G$ of radius $2r$ and degree $c$. 
\end{theorem}

Hence, by combining~\cref{thm:zhu} and~\cref{thm:nc}, we obtain $r$-neighborhood
covers of radius at most $r$ and degree at most $f(r)$ for every class of
bounded expansion. 
We will show in \cref{sec:distributed} how to compute
the $r$-neighborhood covers presented in \cref{thm:nc} in a distributed setting.


\section{Approximating 
dominating sets}\label{sec:approx}

Our first result shows how to (sequentially) compute good \mbox{distance-$r$} dominating sets in any fixed class 
of bounded expansion. 
%
%
The remainder of this section is devoted to the proof of the following theorem. 
\begin{theorem}\label{thm:approxrdom}
For every class $\CCC$ of bounded expansion there is a function $c\colon \N\rightarrow \N$ 
and a linear time algorithm which on input $G\in \CCC$ and $r\in \N$ 
computes an order $L\in \Pi(G)$ witnessing $\wcol_{2r}(G)\leq c(r)$ 
and a $c(r)$-approximation of a minimum distance-$r$ dominating set of~$G$.
\end{theorem}

Our result improves the following result of Dvo\v{r}\'ak~\cite{dvovrak13} who
proved that there exists a $c(r)^2$-approximation, where $c(r)$ 
is the constant in \cref{thm:approxrdom}. 

\ifthenelse{\boolean{IsProc}}{\bigskip}{\bigskip}
Recall the definition of $X_v$ (see~\cref{sec:prelim}):
\begin{equation}
X_v = \{ w\in V(G) \mid  v\in \Wreach_{2r}[G,L,w]\}\,.\label{def:Xv}
\end{equation}
We define for a fixed vertex $v\in V(G)$,
\[R_v \coloneqq \{w \in X_v \mid  v = \min
\Wreach_r[G,L,w]\}\,.\]

\newcounter{lemma-Rv}
\setcounter{lemma-Rv}{\value{theorem}}
\begin{lemma}\label[lemma]{lem:R_v}
  For all vertices $v\in V(G)$ and for all $w\in R_v$ we have $N_r[w] \subseteq X_v$.
\end{lemma}
\newcommand{\prfRv}{%
\begin{proof}
Let  $w\in R_v$. Observe first that $v$ is the minimum element in $N_r[w]$. 
Now let $u\in N_r[w]$. 
Then there is a path of length at most~$r$ from $u$ to $w$ with all 
vertices in $N_r[w]$ and a path of length at most~$r$ from $w$ to $v$ again with 
all vertices in $N_r[w]$. Hence there is a path of length at most~$2r$
from $u$ to $v$ with all vertices in $N_r[w]$. As $v$ is minimal in $N_r[w]$, 
this path witnesses that $v\in \Wreach_{2r}[G,L,u]$, and hence by definition of $X_v$
it holds that $u\in X_v$. 
\end{proof}%
}
\prfRv

\begin{proof}[Proof of \cref{thm:approxrdom}]
We claim that the 
set 
\begin{equation}
D
 \coloneqq \{\min \Wreach_r[G,L,w] \mid 
 w\in V(G)\} = \{v\in V(G) \mid R_v\neq \emptyset\}
\label{eq:D}
\end{equation}
is a $c(r)$-approximation of a minimum distance-$r$ dominating set. Obviously, $D$ is a distance-$r$ dominating set of $V(G)$, as every vertex~$w$ is
dominated by $\min \Wreach_r[G,L,w]$. It remains to show that we achieve the claimed approximation
ratio. 

For $v\in V(G)$, let $X_v$ and $R_v$ be as above. Let $\mathcal{X}$ be the
collection $\{X_v \mid v\in V(G)\}$ as in \Cref{thm:nc}.  Then $\mathcal{X}$ is
an $r$-neighborhood cover of degree $c(r)$ and by \Cref{lem:R_v}, for $w\in R_v$,
we have $N_r[w]\subseteq X_v$.\ifthenelse{\boolean{IsProc}}{ }{

}%
 Let $M$ be a minimum distance-$r$ dominating set in $G$. As every $w\in V(G)$ can be 
distance-$r$ dominated only from $N_r[w] \subseteq X_v$, it follows that if $w\in R_v$, then $M\cap X_v\neq \emptyset$. 

Hence, as every vertex appears in at most $c(r)$ clusters, it holds that
\begin{equation*}
\abs{D} =\abs{\bigcup_{\stackrel{v\in V(G)}{R_v\neq \emptyset}}\{v\}}\leq 
  \sum_{\stackrel{v\in V(G)}{R_v\neq \emptyset}}\abs{M\cap X_v}\leq c(r)\cdot \abs{M}.
\end{equation*}

We finally show how to compute $D$ in linear time. 
We assume that $G$ is stored in memory by $n$ adjacency 
lists. We first use the linear time approximation algorithm for the weak 
coloring numbers from \cref{crl:crl-orientation} to compute an order $L$, which will be
represented in a way such that one can iterate through the vertices along $L$ in
$\Oof(|V(G)|)$ time and such that a comparison $u<_L w$ for every pair $u,w\in
V(G)$ takes $\Oof(1)$ time. Here we assume that vertices are equipped with 
constant size identifiers representing the order. 
The algorithm is given as \Cref{alg:DomSet}. 

\begin{algorithm}
  \caption{$\mathrm{DomSet}(G,L)$}
  \label{alg:DomSet}
  \begin{algorithmic}[1]
    \Require Graph $V(G)$; $V(G) = \{v_i : 1\le i\le n\}$; $ v_i<_Lv_j$ for
    $i<j$; $A(v_i)$ is the adjacency list of $v_i$ 
    \Ensure A $c(r)$-approximation of a dominating set of $G$
    \State\Call{$\mathrm{SortLists}$}{$L$}\hspace{4cm}\Comment{\Cref{alg:SortLists}}\label{line:SortLists}
    \State $D \gets \emptyset$ 
    \State $Dominated \gets \emptyset$ 
    \For{$i\gets 1,\ldots, n$} 
        \State $N_i\gets $\Call{$\mathrm{BFS}$}{$v_i,L$}\label{line:BFS}\hspace{2cm}\Comment{\Cref{alg:BFS}}
        \If{$N_i\setminus Dominated \neq \emptyset$}
          \State $D\gets D\cup \{v_i\}$
          \State $Dominated \gets Dominated \cup N_i$
        \EndIf
    \EndFor
    \State{\textbf{return} $D$}
  \end{algorithmic}
\end{algorithm}
In the first step (\Cref{line:SortLists} of \Cref{alg:DomSet} and
\Cref{alg:SortLists}), we ensure in linear time that every adjacency list is
sorted increasingly with respect to~$L$.  \Cref{alg:SortLists} iterates through
the vertices of~$G$ in order $L$ starting from the least vertex (such that the next vertex can always be found
in constant time) and thus has running time $\Oof(m)$. As graphs of bounded
expansion are degenerate, we have $m \in \Oof(n)$.

\begin{algorithm}
  \caption{$\mathrm{SortLists}(L)$}
\label{alg:SortLists}
  \begin{algorithmic}[1]
    \Require Graph $V(G)$; $V(G) = \{v_i : 1\le i\le n\}$; $ v_i<_Lv_j$ for $i<j$; $A(v_i)$ is the adjacency list of $v_i$ 
    \Ensure  $A(v_i)$ is increasingly sorted with respect to $L$
    \For{$i\gets 1,\ldots,n$}
      $B(v_i) \gets A(v_i)$ and $A(v_i)\gets ()$
      \EndFor
    \For{$i\gets 1,\ldots,n$}
      \For{$v_j \in B(v_i)$}
      \State add $v_i$ at the end of $A(v_j)$
      \EndFor
      \EndFor
  \end{algorithmic}
\end{algorithm}

Now \Cref{alg:DomSet} iterates through the vertices of $G$ starting with the
least element $v_1$ along $L$. For every $v_i\in V(G)$, it uses \Cref{alg:BFS}
to compute the set of vertices that are bigger than $v_i$ and are dominated
by~$v_i$. If such a vertex is not dominated by a vertex smaller than $v_i$, it
serves as a vertex $w$ in the definition of~$D$ (see \cref{eq:D}). Indeed,
$v_i \in \Wreach_r[G,L,w]$ because $w$ was found by a breadth-first search from
$v$ restricted to vertices greater than $v$ with respect to~$L$ and to distances
at most $r$. On the other hand, if not $v$ was
the minimum vertex in $\Wreach_r[G,L,w]$ but, say, $u$, then $w$ would be
dominated by $u$ and added to the set $Dominated$ in the earlier iteration $j$
for $u =v_j$. Thus $v_i = \min \Wreach_r[G,L,w]$ and $v_i$ is
added to~$D$.

\begin{algorithm}
  \caption{$\mathrm{BFS}(v,L,r)$}
  \label{alg:BFS}
  \begin{algorithmic}[1]
    \Statex\LineComment{As usual BFS, but restricted to vertices $w>_Lv$ and
      only for $r$ steps.}
    \Require Graph $G$; for $w\in V(G)$, the adjacency list $A(w)$ is increasingly sorted w.r.t.\@ to
      $L$ 
    \State $Q\gets$ empty queue 
    \State Enqueue $(v,0)$ in $Q$
    \While{$Q\neq \emptyset$} 
      \State $(w,dist)\gets$ dequeue from $Q$
      \If{$dist < r$}
        \For{$u \in A(w)$ and $u>_L v$} \label{line:iter_through_adj_list}
          \If{$u$ not marked as visited}
            \State mark $u$ as visited 
            \State enqueue $(u,dist+1)$ in $Q$
          \EndIf
        \EndFor
      \EndIf
    \EndWhile
    \State{\textbf{return} marked vertices}
  \end{algorithmic}
\end{algorithm}

Let us estimate the running time of \Cref{alg:DomSet}.  Recall that
\Cref{line:SortLists} has linear running time. Note that every set
$N_i$ computed in \Cref{{line:BFS}} for a vertex $v_i$ is a subset of $X_{v_i}$
because \Cref{alg:BFS} restricts its search to vertices bigger than $v_i$ and to
distances at most~$r$. That is, if $w\in N_{v_i}$, then
$v_i\in \Wreach_r[G,L,w] \subseteq \Wreach_{2r}[G,L,w]$ and thus $w\in
X_{v_i}$.
As every graph $G\in \CCC$ is $c(r)$-degenerate, every induced
subgraph $H\subseteq G$ has at most $c(r)\cdot |V(H)|$ many edges, also the graph
induced by $N_i$. When constructing~$N_i$ in \Cref{alg:BFS} we will only visit vertices of
$N_i$ and, for every $w\in N_i$ at most one vertex in its adjacency list that is
not in $N_i$. This can be achieved if \Cref{line:iter_through_adj_list} of
\Cref{alg:BFS} is implemented as an iteration through $A(w)$ starting from the
biggest vertex and stopping if a vertex $u\in A(w)$ with $u<_Lw$ is reached
(recall that $A(w)$ sorted). Hence
this search requires time at most
$\Oof((c(r)+1)\cdot |N_{v_i}|) = \Oof(c(r)\cdot |X_{v_i}|)$. As every vertex $w$ appears
in at most~$c(r)$ clusters $X_v$, we obtain
a running time of $\sum_{v\in V(G)} \Oof(c(r)\cdot |X_v|)=\Oof(c(r)^2\cdot n)$.
\end{proof}

Note that by \cref{thm:zhu} the constant $c(r)$ in the theorem exists for 
every class of bounded expansion. Besides the improved approximation ratio, 
our algorithm is simpler than that of~\cite{dvovrak13}. In particular, given an 
order $L\in\Pi(G)$, it can straightforwardly be implemented in a
distributed way.

\section{Distributed \textbf{\textit{r}}-neighborhood covers and \textbf{\textit{r}}-dominating sets}\label{sec:distributed}

In this section we will show how to compute sparse 
$r$-neigh\-bor\-hood covers as described in \cref{thm:nc} and the \textsc{Distan\-ce-$r$ Dominating Set} 
of \cref{thm:approxrdom} in a distributed setting.

\ifthenelse{\boolean{IsProc}}%
{}%
{\bigskip} In order to compute $r$-neighborhood covers according to \cref{thm:nc},
we want to compute an order~$L$ of $V(G)$ which witnesses that
$\wcol_{2r}(G)\leq c \coloneqq c(2r)$. In the distributed setting, that means that every vertex $w$
learns its weak reachability set $\Wreach_{2r}[G,L,w]$ and, for each $v\in
\Wreach_{2r}[G,L,w]$, a path within $X_v$ of length at most $2r$ from $w$ to $v$.

In order to find the distance-$r$ dominating set described in~\cref{thm:approxrdom}, 
every vertex $w$ will choose as its dominator the vertex $\min \Wreach_r[G,L,w]$
and send a message to that vertex along the stored path. \ifthenelse{\boolean{IsProc}}{}{(Note
that we computed the order $L$ for the parameter $2r$, but are using it for $r$.)}
Even if all vertices send their messages at once, no vertex will have to 
forward more than~$c$ messages.

\ifthenelse{\boolean{IsProc}}%
{}%
{\bigskip} First, using \cref{crl:compute-distributed-wcol}, 
we compute for a given input graph $G$ an order
$L$ witnessing that \mbox{$\wcol_{2r}(G)\leq c$} using $\Oof(r^2\log n)$
communication rounds. Note that the number of rounds does not depend
on $c$, which influences only the size of messages. The latter is
$\Oof(c^2\cdot r\cdot \log n)$, which implicitly follows from
\cref{crl:compute-distributed-wcol} in~\cite{nevsetvrildistributed}.
We show that every vertex can learn its weak reachability set as 
well as a routing scheme which preserves 
short distances. Recall from \Cref{def:Xv} on \Cpageref{def:Xv} that $X_v$ is
defined as $X_v=\{w\in V(G) : v\in \Wreach_{2r}[G,L,w]\}$. 

\newcounter{lemma-orientation}
\setcounter{lemma-orientation}{\value{theorem}}
\setcounter{theorem}{\value{lemma-orientation}}
\begin{lemma}\label[lemma]{lem:orientation}
  Let $\CCC$ be a class of bounded expansion and $r\in \N$. There is a
  constant $c=c(2r)$ such that for every $G\in \CCC$ there is a linear order $L$
  on $V(G)$ such that $|\Wreach_{2r}[G,L,w]|\leq c$ for all $w\in V(G)$ and in
  $\Oof(r^2\cdot \log n)$ communication rounds (in \congestbc) every vertex $w$ can learn
  $\Wreach_{2r}[G,L,w]$ and for each $v\in\Wreach_{2r}[G,L,w]$ a path $P_{v,w}$
  of length at most $2r$ from $w$ to $v$, which is a shortest path between $v$
  and $w$ in the graph induced by $X_v$.  In particular, if
  $v=\min \Wreach_r[G,L,w]$, then the path $P_{v,w}$ is a shortest path between
  $v$ and~$w$ in $G$.
\end{lemma}

\begin{proof}
  The pseudocode is given in \Cref{alg:distr-WReach}.
  First, using~\cref{crl:compute-distributed-wcol}, we compute for a given graph
  $G$ an order~$L$ witnessing that $\wcol_{2r}(G)\leq c(2r)$ in
  $\Oof(r^2\log n)$ communication rounds.

  \newcommand{\Pw}{\mathcal{P}_w} \newcommand{\pth}{\mathrm{path}}
  \newcommand{\NewPaths}{\mathcal{P}} \newcommand{\sid}{\mathrm{sid}}

  The procedure implicitly uses an algorithm of Barenboim and
  Elkin~\cite{barenboim2010sublogarithmic}, which assigns to each vertex~$v$ a
  \emph{class id} $cl(v)$, which together with the unique vertex identifier
  induces the linear order~$L$.  For ease of presentation, we write $v_i$ for
  the vertex at position $i$ in the order $L$ and call $i$ the \emph{super-id}
  of the vertex.

  \begin{algorithm}
    \caption{$\mathrm{\Wreach{}Dist}(r)$}
    \label{alg:distr-WReach}
    \begin{algorithmic}[1]
      \renewcommand{\algorithmicrequire}{\textbf{Input (for a vertex $w$):}}
      \Require $n=\abs{V(G)}$, adjacency list of $w$, id of $w$
      \Ensure{$\Wreach_{2r}[G,L,w]$, $\min \Wreach_r[G,L,w]$ for a particular
        linear order $L$ (see text)} 
      \State in parallel, compute $L$\Comment{by \Cref{crl:compute-distributed-wcol}}
      \Statex\Comment{when done, every vertex $w$ knows
                   its super-id $\sid(w)$} 
      \State in parallel, $\Pw = \bigl\{\{\sid(w)\}\bigr\}$ 
      \For{$i = 1,\ldots,r$} 
        \State in parallel, broadcast $\Pw$, receive new paths in $\NewPaths$ 
        \State in parallel (for vertex $w$): \State $toSend \gets \emptyset$ 
        \For{$u_1$ first vertex in a path from $\NewPaths$} 
          \If{$\sid(u_1) < \sid(w)$}
            \State $P\gets$ shortest path from $\NewPaths \cup \Pw$ that starts 
                        in $u_1$, break ties using super-ids
          \EndIf
          \If{exists $P'\in \Pw$ with $P'=u_1,\ldots$} 
            \State remove $P'$ from $\Pw$
          \EndIf
          \State $\Pw\gets \Pw\cup \{P\}$\Comment{let
                                  $P = u_1,\ldots,u_j$}
          \State $toSend \gets toSend\cup \{u_1,\ldots,u_j,w\}$
      \EndFor
      \State broadcast the set $toSend$
      \EndFor
    \end{algorithmic}
  \end{algorithm}

  The remaining part of the computation has $2r$ rounds which correspond to $2r$
  rounds of a breadth-first search as in \Cref{alg:BFS}. This time, the search
  is performed in parallel and we have to make sure that only a logarithmic
  amount of information is sent by every vertex for the \congestbc model. The
  idea is that every vertex~$w$ forwards only information about paths that start
  in a vertex $v\in\Wreach_{2r}[G,L,w]$.

  Every vertex $w$
  maintains a set $\mathcal{P}_w$ of paths of length at most $2r$ from vertices
  $v\in \Wreach_{2r}[G,L,w]$. For every vertex $v\neq w$ there is at most one
  path $P_v$ in $\mathcal{P}_w$ that starts in~$v$ and certifies that
  $v\in \Wreach_{2r}[G,L,w]$. In the first round, every vertex broadcasts its
  super-id, which we understand as a path of length $0$. A vertex $w$ receives
  super-ids and stores only those which are smaller than its own super-id.

  In a later iteration, every vertex~$w$ receives some sets of paths from its
  neighbors and computes their union $\mathcal{P}$. For every vertex $u_1$, with
  a super-id greater than the super-id of $w$, all paths from~$\mathcal{P}$
  starting in $u_1$ are discarded. For every vertex~$u_1$ with a smaller
  super-id, vertex $w$ selects the shortest path starting in $u_1$ among all
  paths in $\mathcal{P}$ and $\mathcal{P}_w$. (There is at most one path in
  $\mathcal{P}_w$ that starts in $u_1$.) If there are many such shortest paths,
  $w$ chooses the lexicographically least one (with respect to the
  super-ids). Let this path be $P= u_1,\ldots,u_j$ for some $j\le 2r$. Then $P$
  is stored in $\mathcal{P}_w$ (if there is already a path in $\mathcal{P}_w$
  that starts in $u_1$, it is replaced by $P$). If $j<2r$, then $w$ broadcasts
  the path $u_1,\ldots, u_j,w$.

  Observe that every vertex $w$ forwards information about a vertex $v$ only if
  $v\in \Wreach_{2r}[G,L,w]$. Hence,~$w$ forwards only at most $c$ paths
  simultaneously and the whole procedure works in the \congestbc model. Observe
  also that we perform a breadth-first search through the cluster $X_v$ and
  break ties according to the order by vertex super-ids. This implies our claims
  on shortest paths.
\end{proof}

We can now combine \cref{thm:nc} and \cref{lem:orientation} to obtain the first main theorem of this section.

\begin{theorem}
Let $\CCC$ be a class of bounded expansion. There is a distributed algorithm which 
for every graph $G\in \CCC$ and every $r\in \N$ computes a representation
of a sparse $r$-neighborhood cover 
in the \congestbc model in 
$\Oof(r^2\cdot \log n)$ communication rounds. More precisely, the algorithm
computes an order $L$, represented by $\log n$-sized labels and for every
vertex~$v$ a routing scheme of length at most $2r$ to every vertex in 
$\Wreach_{2r}[G,L,v]$. 
\end{theorem}

Also \cref{thm:approxrdom} can now be implemented as a distributed algorithm. 

\begin{theorem}\label{thm:ds}
Let $\CCC$ be a class of bounded expansion and let \mbox{$r\in \N$}. There is a
constant $c(r)$ and a distributed
algorithm which for every graph $G\in \CCC$ computes a $c(r)$-approximation of a
minimum distance-$r$ dominating set in the \congestbc model in $\Oof(r^2\cdot \log n)$ 
communication rounds. 
\end{theorem}
\begin{proof}
Recall that we want to compute the distance-$r$ dominating set
\begin{equation*}
	D \coloneqq \{v \in V(G) \mid  v = \min \Wreach_r[G,L,w] \text{ for some }w\in V(G)\},
\end{equation*}
that is, every vertex $w$ elects the smallest vertex from its $r$-neigh\-bor\-hood with 
respect to $L$ to the distance-$r$ dominating set. As $w$ knows $\Wreach_r[G,L,w]$ and a 
routing scheme to these vertices,
all vertices can send to the smallest vertex in the list a short message
that it should be included in the dominating set. Observe that if a 
vertex~$u$ has to forward the identifier of a vertex $w\in \Wreach_r[G,L,v]$
from some other vertex $v$, then also $w\in \Wreach_r[G,L,u]$. Hence, no
vertex has to forward more than $c(r)$ messages of total size at most 
$\Oof(c(2r)^2\cdot r\cdot \log n)$: every vertex $v$ forwards id's of at
most $c(2r)$ vertices $w$ from its $\Wreach_{2r}[G,L,v]$ together with
their rooting schemes (of size at most $2r$) to at most $c(2r)$
vertices in $\Wreach_{2r}[G,L,w]$.
\end{proof}

%
%
%

\section{Connected Dominating Sets}\label{sec:connected-ds}

In this section we study the \textsc{Connected Distance-$r$ Dominating Set}
problem. 
Our main result in this section is the following theorem.
\begin{theorem}\label{thm:conn_ds_congest}
Let $\CCC$ be a class of bounded expansion and let $r\in \N$. There is a
constant $c$ and a distributed
algorithm which for every graph $G\in \CCC$ computes a $c$-approximation of a
minimum connected distance-$r$ dominating set in the \congestbc model in $\Oof(r^2\cdot \log n)$ 
communication rounds. 
\end{theorem}

The following observation is folklore.

\newcounter{lemma-r-conn}
\setcounter{lemma-r-conn}{\value{theorem}}
\begin{lemma}\label[lemma]{lem:r-conn}
Let $G$ be a connected graph and let $D$ be a \mbox{distance-$r$} dominating set of $G$. Let~$\mathcal{P}$ be a set of paths in $G$ such that for each pair
$u,v\in D$ with $\mathrm{dist}(u,v)\leq 2r+1$ there is a path 
$P_{u,v}\in \mathcal{P}$ connecting $u$ and $v$. Then the subgraph $H$
induced by $D\cup \bigcup_{P\in \mathcal{P}}V(P)$ is connected. 
\end{lemma}%
\newcommand{\prfLemmaRConn}{%
\begin{proof}
We show by induction on $\mathrm{dist}(u,v)$ that all $u,v\in D$ are
connected in $H$. The claim holds by definition of $H$ if 
$\mathrm{dist}(u,v)\leq 2r+1$. Now assume that $\mathrm{dist}(u,v)\geq 2r+2$ 
and let $P=(u=v_0,v_1,\ldots, v_t=v)$ with $t\geq 2r+3$ be a shortest path 
connecting $u$ and $v$. As $P$ is a shortest path, neither~$u$ nor $v$ 
dominate $v_{r+1}$. Hence there is another vertex $w\in D$ 
which dominates $v_{r+1}$. As $v_1$ and $v_{r+1}$ are connected by a path
of length $r$, and $w$ and $v_{r+1}$ are connected by a path $(w=w_0,w_1,\ldots, w_{r'}=v_{r+1})$ of length $r'\leq r$,
$\mathrm{dist}(u,w)\leq 2r$, hence $w$ and $u$ are connected in $H$. 
Furthermore, the path $P'=(w, w_1,\ldots, w_{r'-1},v_{r+1},\ldots, v_t=v)$ is 
shorter than 
$P$, hence, by induction hypothesis, $w$ and $v$ are connected in $H$. 
This implies that $u,v$ are connected in~$H$. 
\end{proof}%
}%
\ifthenelse{\boolean{IsProc}}{}{\prfLemmaRConn}

Now, we use the local separation properties of the weak colouring
numbers to connect the dominating set we computed in \cref{thm:ds}. 
The proof of the following lemma is immediate by definition of weak reachability. 

\begin{lemma}\label{lem:decomposepath}
Let $G$ be a graph and let $L$ be a linear order on $V(G)$. 
Let $u,v\in V(G)$ be such that there exists a path $P$ between $u$ and
$v$ of length at 
most $r$. Let $w$ be the minimal vertex of $P$ with respect to
$L$. Then $w\in \Wreach_r[G,L,u]$ and $w\in \Wreach_r[G,L,v]$.
\end{lemma}

\ifthenelse{\boolean{IsProc}}{From \Cref{lem:r-conn} and
  \Cref{lem:decomposepath} we obtain the following corollary.}{}
\newcounter{corollary_conn_ds}
\setcounter{corollary_conn_ds}{\value{theorem}}
\begin{corollary}\label{crl:conn_ds}
Let $G$ be a connected graph and let $L$ be a linear order on $V(G)$. 
Let $D$ be an $r$-dominating set of $G$. Let $D'$ be a set which
is obtained by adding for each $v\in D$ and each $w\in \Wreach_{2r+1}[G,L,v]$
the vertex set of a path between $v$ and $w$. Then $D'$ is a connected 
distance-$r$ dominating set of $G$. 
\end{corollary}
\begin{proof}
Fix a set
$\mathcal{P}$ of paths in $G$ such that for each pair
$u,v\in D$ with $\mathrm{dist}(u,v)\leq 2r+1$ there is a path 
$P_{u,v}\in \mathcal{P}$ connecting $u$ and $v$. According to
\cref{lem:r-conn}, the subgraph $H$
induced by $D\cup \bigcup_{P\in \mathcal{P}}V(P)$ is connected. 
According to \cref{lem:decomposepath}, for each path $P_{u,v}$
between $u$ and $v$, there is a vertex $w\in V(P_{u,v})$ is
weakly $2r+1$-reachable both from $u$ and from $v$. As $D'$ 
contains the vertex set of a path between $u$ and $w$ and of
a path between $v$ and $w$, it follows that $D'$ is a connected
distance-$r$ dominating set of $G$. 
\end{proof}%

We are now ready to prove the main theorem. 

\begin{proof}[Proof of~\cref{thm:conn_ds_congest}]
  Instead of computing an order $L$ for $\wcol_{2r}(G)$ as in \cref{thm:ds}, we
  compute an order $L$ for $\wcol_{2r+1}(G)$.  Assume
  $|\Wreach_{2r+1}[G,L,v]|\leq c'$ for all $v\in V(G)$.  We compute an
  $r$-dominating set $D$ based on the order~$L$.  Note that in
  \Cref{sec:distributed} we used $L$ computed for parameter $2r$ and now we use
  $L$ computed for $2r+1$, but for all orders $L$ and all $v\in V(G)$ we have
  $\abs{\Wreach_{2r}[G,L,v]} \le \abs{\Wreach_{2r+1}[G,L,v]}$.

By~\cref{thm:approxrdom}, the set $D$ is at most
$c'$ times larger than a minimum distance-$r$ dominating set. 
As a by-product, see \cref{lem:orientation}, every vertex $v$ 
learns a path of length at most $2r+1$ to each 
$w\in \Wreach_{2r+1}[G,L,v]$. Now, every vertex broadcasts
its set of paths to construct the set $D'$. As in the proof 
of~\cref{thm:ds}, observe that if a 
vertex $x$ has to forward a path from $w\in \Wreach_r[G,L,v]$
to $v$ for some other vertex $v$, then also $w\in \Wreach_r[G,L,x]$. Hence, no
vertex has to forward more than $c'$ messages of total size at most 
$\Oof(c'\cdot r\cdot \log n)$. Clearly, the computed set $D'$ has
size at most $c'\cdot (2r+1)\cdot |D|$ and by 
\cref{crl:conn_ds} it is a connected distance-$r$ dominating set. 
We conclude be defining $c\coloneqq c'^2\cdot (2r+1)$. 
\end{proof}

We now show how to use the greater power of the \local model 
to compute connected dominating sets with much smaller constants involved. 
Our theorem is based on the simple observation that in the \local
model we can construct for every connected graph from an $r$-dominating set~$D$ 
a connected depth-$r$ minor with $|D|$ vertices. This minor (by definition of 
bounded expansion classes) has only a linear number of edges and we can hence 
choose a set of short paths realizing
the corresponding connections to connect the dominating set. 

\newcommand{\lex}{\mathrm{lex}}
We want to define a partition of $V(G)$ into balls around vertices from an
r-dominating set~$D$.
For a connected graph $G$ and an injection $id\colon V(G) \to \N$, we define the
\emph{lexicographic order} $<_\lex$ on the set of paths in $V(G)$ with respect
to $id$ as
follows. Consider two paths $P_1 = v_1,\ldots,v_k$ and $P_2 = w_1,\ldots,w_\ell$. If $k<\ell$,
then $P\le_\lex P_2$. If $k=\ell$, then $P_1 \le_\lex P_2$ if the sequence
$id(v_1),\ldots,id(v_k)$ is lexicographically smaller than the sequence
$id(w_1),\ldots,id(w_\ell)$ or $P_1=P_2$. For vertices $v,w\in V(G)$, let $P(v,w)$ be the
lexicographically shortest path from $v$ to $w$.

Let $G$ be a connected graph, let $id(v)$ be the unique identifier of $v$ and
let $D$ be a distance-$r$ dominating set of $G$. For each $v\in D$ let
\[B(v)\coloneqq \{w \in V(G) \mid P(v,w)\le_\lex P(u,w) \text{ for all } u\in
D, u\neq v\}\,.\] The \emph{$D$-partition} $\mathcal{B}(D)$ of $G$ with respect
to $id$ is the set $\{B(v) \mid $ $v\in D\}$.


\newcounter{lemma-rad-r}
\setcounter{lemma-rad-r}{\value{theorem}}
\begin{lemma}\label{lem:rad-r}
Let $G$ be a connected graph and let $D$ be a distance-$r$ dominating set of $G$. 
Then $\mathcal{B}(D)=\{B(v) : v\in D\}$ is a partition of 
$V(G)$ and $G[B(v)]$ has radius at most $r$ for all $v\in D$. 
\end{lemma}
\begin{proof}
As $G$ is connected and $D$ is a distance-$r$ dominating set, $\mathcal{B}$ is
a partition of $V(G)$. Furthermore, for each $w\in V(v)$, there is a lexicographically 
shortest path $P$ of 
length at most $r$ from~$v$ to~$w$ in~$G$. Assume towards a contradiction 
that $P$ is not also a path in $B(v)$. Then there is $z\in V(P)$ and $u\in D$ such 
that $z\in B(u)$. By definition of $B(u)$, the lexicographically shortest path $Q'$ from
$u$ to~$z$ is smaller than the lexicographically shortest path $Q$ between $v$ and $z$. 
But then the path $P'$ obtained by replacing the initial part~$Q$ of $P$ by $Q'$ is 
lexicographically smaller than $P$, a contradiction. 
\end{proof}%

\begin{lemma}\label{lem:minor-from-D}
Let $G$ be a connected graph and let $D$ be a distance-$r$ dominating set of $G$. 
By contracting the sets $B(v)$ for $v\in D$, we obtain a connected depth-$r$ minor of $G$. 
\end{lemma}
\begin{proof}
It is immediate by definition of depth-$r$ minors and \cref{lem:rad-r} that we 
construct a depth-$r$ minor $H\minor_rG$. Furthermore, as $G$ is connected and 
as $\mathcal{B}(D)$ is a partition of $V(G)$ by the same lemma, 
it is easy to see that $H$ is connected.
\end{proof}%

\newcounter{lemma_LOCAL}
\setcounter{lemma_LOCAL}{\value{theorem}}
\begin{lemma}\label{lem:LOCAL}
Let $G$ be a connected graph such that for each depth-$r$ minor $H\minor_r G$
we have $\abs{E(H)}\leq d\cdot \abs{V(H)}$. Let $D$ be a distance-$r$ dominating set of $G$. 
We can compute a connected dominating set $D'$ of $G$ of size at most $2r\cdot d\cdot \abs{D}$
in $3r+1$ communication rounds in the \local model. 
\end{lemma}
\begin{proof}
In this proof, we write $H(D)$ for the depth-$r$ minor constructed from a
distance-$r$ dominating set $D$ as in~\cref{lem:minor-from-D}. 

Every vertex $v\in D$ can find its $2r+1$-neighborhood in $2r+1$ communication 
rounds. With this information, each $v\in D$ can construct $B(v)$, as all possible
dominators for $w\in N_r[v]$ must come from $N_{2r}[v]$. Each vertex $v\in D$
(now understood as representing a vertex of $H(D)$) can also learn its neighbors
in $H(D)$ (here we need to learn $N_{2r+1}[v]$). Now each vertex $v$ computes 
the lexicographically shortest
path $P_{uv}$ of length at most $2r+1$ for each neighbor $u$ in $H(D)$ (take the 
ordering induced by vertex id's). Observe that $u$ and $v$ fix the same path 
$P_{vu}$, hence, the two vertices can report to all vertices on $P_{vu}$ in 
another $r$ communication rounds that they shall be included in the connected 
dominating set $D'$. 

By~\cref{lem:minor-from-D}, the constructed set $D'$ is a connected dist\-ance-$r$
dominating set. Furthermore, by assumption, $H(D)$ has at most $d\cdot|D|$ many 
edges. Each edge is replaced by at most $2r-1$ vertices in the above construction. 
Adding the $|D|$ vertices of the original set $D$, we obtain the claimed bounds. 
\end{proof}

As a corollary from \cref{lem:rad-r} and \cref{lem:LOCAL} we obtain the
following theorem.

\begin{theorem}\label{thm:connecteddomset}
Let $\mathcal{C}$ be a class of graphs of bounded expansion and assume that for every
graph $G\in \mathcal{C}$ we can compute a $c$-approximation $D$ of a 
minimum \mbox{distance-$r$} dominating 
set of $G$ in $t$ rounds in the \local model. Let
$f\colon\N\rightarrow \N$ denote the edge density function of
depth-$r$ minors of $\CCC$. Then there is a distributed
algorithm which finds a $2rcf(r)$-approximation for connected distance-$r$ dominating 
set of $G$ in $\Oof(t+r)$ rounds in the \local model.
\end{theorem}

 
The theorem can be applied, e.g., to extend the algorithm of Lenzen et
al.~\cite{lenzen2013distributed} to obtain a connected dominating set on planar
graphs in the local model which is only $6$ times larger than the dominating set
computed for the planar graph (an $n$-vertex planar graph has at most $3n-6$
edges).  Similarly, it applies to the extension of Lenzen et al.'s algorithm by
Amiri et al.~\cite{AmiriSS16} for graphs of bounded genus or to the randomized
$\Oof(a^2)$ approximation of Lenzen and Wattenhofer~\cite{lenzen2010minimum}
applied to graphs with excluded minors (here, $a\in \Oof(t\log t)$ if $K_t$ is
excluded as a minor).

\section{Conclusion}\label{sec:conclusion}

What are the most general classes of graphs that admit 
efficient algorithms for certain problems? The ambitious 
goal to answer this question for the dominating set problem
has lead to strong graph theoretic and algorithmic results
once it was known that it cannot be solved efficiently 
in full generality. Lower bounds both in classical complexity
and in distributed complexity have motivated the study
of more and more general graph classes. 

Bounded expansion classes of graphs are very general 
classes of sparse graphs. In this paper we proposed a 
new constant factor approximation algorithm 
for the \textsc{(Connected) Distance-$r$ Dominating Set} problem on these
classes of graphs. The algorithm 
improves the previously best known approximation algorithm by Dvo\v{r}\'ak~\cite{dvovrak13}, 
however, its main feature is that it is tailored to be implemented in a 
distributed setting. 

It was proved in \cite{bansal2017tight} that on the class of graphs with arboricity at most $a$ the size of a minimum dominating set can be approximated by a factor $3a$ by an LP rounding algorithm, but that 
it is NP-hard to approximate the size of the minimum dominating set  to within $a-1-\epsilon$ in this class for every $\epsilon>0$.
This natural leads to consider that bounded average degree (or bounded arboricity) is a natural requirement for a class of graphs closed under taking subgraphs to allow constant factor approximation for the size of the minimum dominating set.
By considering subdivisions, it follows that the property of every graph present as an $r$-subdivision in the class should have average degree at most $C(r)$ --- that is of having bounded expansion --- is a natural requirement to allow constant factor approximation for the size of the distance-$r$ minimum dominating sets for every $r$.


Our techniques are based on a distributed computation of sparse neighborhood
covers of small radius on bounded expansion classes of graphs. 
Formerly, no distributed algorithms that compute such covers 
were known and we believe 
that these techniques are interesting beyond the presented 
applications of computing (connected) dominating sets. 
We pose the question whether sparse neighborhood covers can be
computed in distributed constant time. This question is open even 
on more restrictive graph classes, e.g., on planar graphs, 
where dominating sets can be approximated in constant time.

\end{document}